\documentclass[11pt]{article}
\usepackage{amsmath,amssymb,amsfonts} 
\usepackage{epsfig} \usepackage{latexsym,nicefrac,bbm}
\usepackage{xspace}
\usepackage{color,fancybox,graphicx,subfigure,fullpage}
\usepackage[top=1in, bottom=1in, left=1in, right=1in]{geometry}
\usepackage{tabularx} \usepackage{hyperref} 
\usepackage{pdfsync}
\usepackage[boxruled]{algorithm2e}
\usepackage{multicol}
\newcommand{\parag}[1]{ {\bf \noindent #1}}

\newcommand{\nfrac}{\nicefrac}

\newcommand{\conv}{\mathrm{conv}}

\newcommand{\argmax}{\operatornamewithlimits{argmax}}

\newcommand{\cM}{\mathcal{M}}
\newcommand{\cB}{\mathcal{B}}

\newcommand{\cI}{\mathcal{I}}

\newcommand{\OPT}{\mathrm{OPT}}

\newcommand{\st}{\mathrm{s.t.}}

\newenvironment{proof}{\noindent{\bf Proof:}\hspace*{1em}}{\qed\bigskip}
\newcommand{\qed}{\hfill\ensuremath{\square}}

\def\showauthornotes{0} 
 
\def\showdraftbox{0}

\newtheorem{theorem}{Theorem}[section]

\newtheorem{definition}{Definition}[section]
\newtheorem{lemma}[theorem]{Lemma}

\newtheorem{fact}[theorem]{Fact}

\def\FullBox{\hbox{\vrule width 6pt height 6pt depth 0pt}}

\def\qed{\ifmmode\qquad\FullBox\else{\unskip\nobreak\hfil
\penalty50\hskip1em\null\nobreak\hfil\FullBox
\parfillskip=0pt\finalhyphendemerits=0\endgraf}\fi}

\def\qedsketch{\ifmmode\Box\else{\unskip\nobreak\hfil
\penalty50\hskip1em\null\nobreak\hfil$\Box$
\parfillskip=0pt\finalhyphendemerits=0\endgraf}\fi}

\newenvironment{proofof}[1]{\begin{trivlist} \item {\bf Proof
#1:~~}}
  {\qed\end{trivlist}}

\newcommand\Z{\mathbb Z}
\newcommand\N{\mathbb N}
\newcommand\R{\mathbb R}

\newcommand{\marginlabel}[1]%
{\mbox{}\marginpar{\it{\raggedleft\hspace{0pt}#1}}}

\definecolor{Mygray}{gray}{0.8}

 \ifcsname ifcommentflag\endcsname\else
  \expandafter\let\csname ifcommentflag\expandafter\endcsname
                  \csname iffalse\endcsname
\fi

\ifnum\showauthornotes=1
\newcommand{\Authornote}[2]{{\sf\small\color{red}{[#1: #2]}}}
\newcommand{\Authoredit}[2]{{\sf\small\color{red}{[#1]}\color{blue}{#2}}}
\newcommand{\Authorfnote}[2]{\footnote{\color{red}{#1: #2}}}
\newcommand{\Authorfixme}[1]{\Authornote{#1}{\textbf{??}}}
\newcommand{\Authormarginmark}[1]{\marginpar{\textcolor{red}{\fbox{
#1:!}}}}
\else
\newcommand{\Authornote}[2]{}
\newcommand{\Authoredit}[2]{}
\newcommand{\Authorcomment}[2]{}
\newcommand{\Authorfnote}[2]{}
\newcommand{\Authorfixme}[1]{}
\newcommand{\Authormarginmark}[1]{}
\fi

\newcommand{\inparen}[1]{\left(#1\right)}             
\newcommand{\inbraces}[1]{\left\{#1\right\}}           
\newcommand{\insquare}[1]{\left[#1\right]}             

\newcommand\pr{\mathop{\mbox{\bf Pr}}}

\def\abs#1{\left| #1 \right|}
\newcommand{\norm}[1]{\ensuremath{\left\lVert #1 \right\rVert}}

\newcommand{\diag}[1]{{\sf Diag}\left({#1}\right)}

 {
	\begin{enumerate}}{\end{enumerate}}

\newcounter{lecnum}

\newlength{\tpush}
\ifnum\showdraftbox=1

\else

\fi

\newcommand{\expect}[1]{\mathbb{E}\insquare{#1}}
\newcommand{\wt}[1]{\widetilde{#1}}

\title{{\bf Subdeterminant Maximization via Nonconvex Relaxations and Anti-Concentration} }

\date{}

\usepackage{authblk}

\author[1]{Javad B. Ebrahimi}
\author[2]{Damian Straszak}
\author[2]{Nisheeth K. Vishnoi}
\affil[1]{\small Sharif University of Technology, Iran}
\affil[2]{\small \'{E}cole Polytechnique F\'{e}d\'{e}rale de Lausanne (EPFL), Switzerland}

\begin{document}
\maketitle

\begin{abstract}
Several fundamental  problems that arise in optimization and computer science can be cast as follows:  Given vectors $v_1,\ldots,v_m \in \mathbb{R}^d$ and a {\em constraint} family $\cB \subseteq 2^{[m]}$, find a set $S \in \cB$ that maximizes the squared volume of the simplex spanned by the vectors in $S$.
A motivating example is the ubiquitous data-summarization problem in machine learning and information retrieval where one is given a collection of feature vectors that represent data such as documents or images. 
The volume of a collection of vectors is used as a measure of their {\em diversity}, and partition or matroid constraints over $[m]$ are imposed in order to ensure resource or fairness constraints. 
Even with a simple cardinality constraint ($\cB={[m] \choose r}$), the problem becomes NP-hard and has received much attention starting with a result by Khachiyan \cite{Kha95} who gave an $r^{O(r)}$ approximation algorithm for this problem.
Recently, Nikolov and Singh \cite{NS16} presented a convex program and showed how it  can be used to  {\em estimate} the value of the most diverse set when there are  multiple cardinality constraints  (i.e., when $\cB$ corresponds to a partition matroid).
Their proof of the integrality gap of the convex program relied on an  inequality by Gurvits \cite{Gurvits06}, and was recently extended to regular matroids  \cite{SV17,AO17}.
The question of whether these estimation algorithms can be converted into  the more useful {\em approximation} algorithms -- {\em that also output a set} --  remained open.

The main contribution of this paper is to give the first approximation algorithms for both partition and regular matroids.  
We present novel formulations for the subdeterminant maximization problem for these matroids; 
 this reduces  them to the problem  of finding a point that maximizes the absolute value of  a nonconvex function over a Cartesian product of probability simplices.
The technical core of our results  is a new anti-concentration inequality for dependent random variables that arise from these functions which allows us to relate the optimal value of these nonconvex functions to their value at a random point.
Unlike prior work on the constrained subdeterminant maximization problem, our proofs do not rely on real-stability or convexity and could be of independent interest  both in algorithms and complexity where anti-concentration phenomena have recently been deployed.

\end{abstract}

\newpage

\tableofcontents
\newpage

\section{Introduction}

A  variety of problems in computer science and optimization can be formulated as the following constrained subdeterminant maximization problem: Given a positive semi-definite (PSD) matrix $L\in \R^{m\times m}$ and a family  $\cB$ of subsets of $[m]:=\{1,2,\ldots, m\}$,  find a set $S \in \cB$ that maximizes $\det(L_{S,S})$ where $L_{S,S}$ is the principal sub-matrix of $L$ corresponding to rows and columns from $S$. 
Equivalently, if $L=V^\top V$ where $V \in \mathbb{R}^{d \times m}$ is a Cholesky decomposition of $L$, and $V_1,\ldots,V_m$ correspond to the columns of $V$, then the problem is to output a set $S \in \mathcal{B}$ that maximizes the squared volume of the parallelepiped spanned by the vectors  $\{V_i:  i\in S\}$. 
If the family $\cB$ is specified explicitly as a list of its members, this optimization problem, trivially, has an efficient algorithm.
The interesting case of the problem is when $|\cB|$ is large (possibly exponential in $m$) and   an efficient implicit representation or an appropriate  separation oracle is given.

This problem, in its various avatars, has received significant attention in  optimization, machine learning, and  theoretical computer science due to its practical importance and  mathematical connections. 
In geometry and optimization, the vector formulation of the subdeterminant maximization problem for the family $\cB={[m] \choose r}$  is related to several volume maximization \cite{GKL95} and matrix low-rank approximation  \cite{GT2001} problems.
In mathematics, the probability distribution on $2^{[m]}$ in which a set $S \subseteq [m]$ has probability $\pr(S)\propto \det(L_{S,S})$ is referred to as a determinantal point process (DPP); see \cite{Lyons02}.
DPPs are important objects of study in combinatorics, probability, physics and, more recently, in computer science as they provide excellent models for diversity in machine learning \cite{KuleszaTaskar12}. 
Here, the constrained subdeterminant maximization problem corresponds to a constrained MAP-inference problem -- that of finding the most probable set from the family $\cB$; see \cite{CDKV16,CDKSV17} for related problems on  DPPs.
Different constraint families can be employed to ensure various priors,  resource, or fairness constraints on the probability distribution.

Algorithmically, 
even the simplest of constraints make the constrained subdeterminant maximization problem NP-hard; for instance,  when $\cB={[m] \choose r}$.
As the set $\cB$ becomes more complicated,  algorithms for the constrained subdeterminant maximization problem roughly fall into two classes: 1) {\em approximation algorithms} that  {\em output a set} $S \in \cB$ such that $\det(L_{S,S})$ is within some factor of the optimal value and, (2) {\em estimation algorithms}  that just {\em output a number} that is within some factor of the optimal value.

Approximation algorithms for the constrained subdeterminant maximization problem are rare;  
\cite{Kha95} proposed the first polynomial-time approximation algorithm for the problem when  $\cB={[m] \choose r}$ which achieved an approximation factor of $r^{O(r)}$ and, importantly, did not depend on the entries of the underlying matrix. 
This result was  improved by \cite{Nikolov15} where an algorithm which achieved an approximation factor of  $e^r$ was presented.
On the other hand, it was shown \cite{DEFM14, CM09} that there exists a constant $c >1$ such that approximating the  $\cB={[m] \choose r}$ case  with approximation ratio better than $c ^r$ remains NP-hard.  

Among estimation algorithms, recently, \cite{NS16} generalized the result of~\cite{Nikolov15} to the setting when the family $\cB$ corresponds to the bases of a {\em partition matroid}. 
They presented an elegant convex program that allowed them to efficiently {estimate} the value of the maximum determinant set from $\cB$ to within a factor of $e^r$ where $r$ is the size of the largest set in the partition matroid $\cB$.
One of the main ingredients in their proof is an inequality due to \cite{Gurvits06} concerning real stable polynomials.
Building  on this work, \cite{SV17,AO17} presented estimation algorithms for large classes of families $\cB$, such as bases of a regular matroid.
While  the results of \cite{NS16,SV17,AO17} made interesting connections between  convex programming, real-stable polynomials and matroids to design estimation algorithms for the constrained subdeterminant maximization problem, 
the question of whether these estimation algorithms can be converted into approximation algorithms   remained open.

Making these approaches constructive is not only crucial for them to be deployed in the practical problems that motivated  their study, mathematically, there seem to be barriers in doing so.
The main contribution of this paper is to present a new methodology to address  the constrained subdeterminant maximization problem that results in approximation algorithms  for partition and regular matroids.
We obtain our results through a synthesis of novel nonconvex formulations for these constraint families with a new anti-concentration inequality. 
Together, they allow for a simple polynomial-time randomized algorithm that outputs a set $S \in \cB$ with high probability.
Approximation guarantees of our algorithms are close to  prior non-constructive results in several interesting parameter regimes.
The simplicity and generality of our results suggest that our techniques, in particular  the anti-concentration inequality and its use in understanding nonconvex functions, are likely to find further applications.

\vspace{-3mm}
\subsection{Overview of our contributions}

\paragraph{Anti-concentration inequality.} We start by describing the common component to both our applications --  an anti-concentration inequality.
We consider multi-variate functions in which each variable is uniformly and independently distributed over a probability simplex.
 Roughly, our anti-concentration inequality says that if the restriction of such a function along each variable has a certain {\em anti-concentration} property then the function is anti-concentrated over the entire domain. 
Formally, the anti-concentration result applies whenever the multi-variate function satisfies the following property.
\begin{definition}[Anti-concentrated functions]\label{def:main}
For $\gamma \geq 1$,  a  nonnegative measurable\footnote{We always assume that the functions we deal with are {\em regular enough}. Formally, we require measurability with respect to the Lebesgue measure.}   function $f: \Delta_{d} \to \R$  is called {\bf $\gamma$-anti-concentrated} if 
for every  $c\in (0, 1)$
$$\pr \left[ f(x) \geq  c \cdot \OPT \right] \geq1- \gamma dc ,$$   
where $x$ is drawn from the uniform distribution over $\Delta_d$ and $\OPT:= \max_{z \in \Delta_d} f(z)$ is the maximum value $f$ takes on $\Delta_d$.\footnote{$\Delta_d$  denotes the standard $(d-1)-$simplex, i.e., $\Delta_d:=\inbraces{x\in \R^d:\sum_{i=1}^dx_i=1, x\geq 0}.$}

Similarly, for any $r \geq 1$ and any $p_1, p_2, \ldots, p_r \geq 0$, a nonnegative function $f: \prod_{i=1}^r \Delta_{p_i} \to \R$  is said to be  $\gamma$-anti-concentrated  if for every coordinate $i\in \{1,2,\ldots, r\}$, and for every choice of $a_j \in \Delta_{p_j}$ for $j\neq i$, the function $x \mapsto f(a_1, \ldots, a_{i-1}, x, a_{i+1}, \ldots, a_r)$ is $\gamma$-anti-concentrated.
\end{definition}
Perhaps one of the simplest examples of an anti-concentrated function is the  univariate map $t\mapsto |at+b|$ over the domain $[0,1]$.
It is not hard to see that it satisfies the condition of Definition~\ref{def:main} for $\gamma=2$ (see  Lemma~\ref{lemma:anti_conc_norm}). 
It  also follows that for every multi-affine polynomial $p\in \R[x_1, x_2, \ldots, x_r]$ the function $x\mapsto |p(x)|$ is $2$-anti-concentrated. 
Another class of functions that satisfy such an anti-concentration property arise by considering norms and volumes in Euclidean spaces; for instance, functions  of the form $t \mapsto \norm{ut + (1-t)v}_2$ for vectors $u,v$.

\begin{theorem}[Anti-concentration inequality]\label{thm:main0}
Let $\gamma\geq 1$ be a constant. Let $r\geq 2$ and $p_1,\ldots,p_r$ be  positive integers. For every $\gamma$-anti-concentrated function $f : \prod_{i=1}^r \Delta_{p_i}\to \mathbb{R}$, if $x$ is sampled from the uniform distribution on $\prod_{i=1}^r \Delta_{p_i}$, then
 
$$
\pr\left[f(x) \geq  (\gamma e^2)^{-r} \cdot \prod_{i=1}^r{\frac{1}{p_i}}\cdot  \OPT\right] \geq \frac{1}{2},
$$
where $\OPT:=\max \{f(z): z \in \prod_{i=1}^r \Delta_{p_i}\}$ is the maximum value $f$ takes on its domain.
\end{theorem} 
Consequently, the value of a $\gamma$-anti-concentrated function  at a random point in its domain gives an  estimate of its maximum value. 
 In the simplest non-trivial case, it applies to multi-affine functions over the hypercube $[0,1]^r$ and says that the value of the function at a random point is at least $c^{-r}$ times its optimal value, with significant probability (where $c>1$ is an absolute constant). 
It is also easy to see that the bound in Theorem \ref{thm:main0} is tight: For $p(x)=\prod_{i=1}^r x_i$, one can show that the probability that $|p(x)| \geq (\nfrac{3}{4})^r$ over a random choice of $x\in [0,1]^r$  is exponentially small.

As an important special case of Theorem \ref{thm:main0}, consider the setting in which $p_i=2$ for $i=1,2, \ldots, r$ (i.e., the domain is the hypercube $[0,1]^r$) and $f(x):=\abs{p(x)}$ where $p\in \R[x_1, \ldots, x_r]$ is a multi-affine polynomial. 
Using the fact noted earlier that such an $f$ is  $2$-anti-concentrated,  we conclude from Theorem~\ref{thm:main0} that  for some absolute constant $c>1$ and a uniformly random choice of $x\in [0,1]^r$,
\begin{equation}\label{eq:anti_poly}
\pr\insquare{|p(x)| \geq c^{-r} \cdot \max_{z \in [0,1]^r} |p(z)| }\geq \frac{1}{2}.
\end{equation}
This gives us a way to {\em estimate} the maximum of  $|p(x)|$ over $[0,1]^r$ by just evaluating it on a certain number of random points and outputting the largest one. 
However, this observation does not directly give us much insight about the problem we typically would like to solve; that of maximizing $|p(b)|$ over binary vectors $b\in \{0,1\}^r$. 
Towards this, note  that for a multi-affine polynomial $p$, 
$$\max_{z \in \{0,1\}^r} |p(z)| =  \max_{z \in [0,1]^r} |p(z)|.$$
Moreover, the above has a simple algorithmic proof that follows from the convexity of $x\mapsto |p(x)|$ restricted to coordinate-aligned lines. This allows us to use the above algorithm to find a point $b\in \{0,1\}^r$ whose value is at most $c^{r}$ times worse than  optimal given only an evaluation oracle for $p$. In particular, no assumptions are made on the analytic properties of $p$, such as concavity or real stability. 
In fact, in most interesting cases, such functions are highly nonconvex, hence standard convex optimization tools do not apply.

\vspace{-4mm}
\paragraph{Partition matroids.}
As a first application of Theorem \ref{thm:main0}, we provide an approximation algorithm for the problem of subdeterminant maximization under partition constraints. 
Let $\mathcal{P}:=\{M_1,M_2, \ldots, M_t\}$ be a partition of $[m]:=\{1,2,\ldots, m\}$ into non-empty, pairwise disjoint subsets and let $b=(b_1,b_2,\ldots,b_t)$ be a sequence of positive integers. 
Then the set 
$\cB := \{S \subseteq [m]: |S\cap M_i|=b_i \mbox{ for all }i=1,2,\ldots, t\}$
is called a partition family induced by $\mathcal{P}$ and $b$. 
We first show that the problem of finding the determinant-maximizing set under partition constraints can be reformulated as
$$\max_{x\in \Delta} \det\inparen{W(x)^\top W(x)}^{1/2}$$
where $\Delta$ is a certain product of simplices, and $W(x)$ is a matrix whose $i$-th column is a convex combination of certain vectors derived from $L=V^\top V$ and the variables in $x$.
Subsequently,  we show that such functions are $2$-anti-concentrated, which allows us to apply Theorem~\ref{thm:main0} to obtain the following result.

\begin{theorem}[Subdeterminant maximization under partition constraints]\label{thm:main2}
There exists a polynomial-time randomized algorithm such that given a PSD matrix $L\in \R^{m\times m}$, a partition $\mathcal{P}=\{M_1,M_2,\ldots ,M_t\}$ of $[m]$ and a sequence of numbers $b=(b_1, b_2, \ldots, b_t)\in \N^t$ with $\sum_{i=1}^t b_i=r$, outputs a set $S$ in the induced partition family $\cB$ such that with high probability
 $$\det(L_{S,S}) \geq  \OPT \cdot  (2e)^{-2r} \cdot \prod_{i=1}^t {\left(\frac{1}{p_i}\right)^{b_i}},$$ 
  where $\OPT:=\max_{S\in \cB} \det(L_{S,S})$ and $p_i:=|M_i|$ for $i=1,2, \ldots, t$. 
\end{theorem}
Prior work by \cite{NS16} outputs a random set whose value is at most $e^r$ times worse than $\OPT$ in {\it expectation} and unlike the theorem above,  does not yield a polynomial-time approximation algorithm, as the probability of success can be exponentially small (see Appendix~\ref{sec:nshard}).
Further, in the case when $p_i=O(1)$ for all $i$ and $b_i=1$ for all $i$ (i.e., when every part has constant size and exactly one vector from every part has to be selected) the approximation ratio of our algorithm is $c^r$ for some constant $c>1$,
 which, up to the constant in the base of the exponent, matches their result. 

\vspace{-2mm}
\paragraph{Regular matroids.} Our second result for the constrained subdeterminant maximization problem is  for the case of regular matroids (i.e., when the constraint family $\cB$ arises as a set of bases of a regular matroid; see Section~\ref{sec:preliminaries}).
 To apply Theorem~\ref{thm:main0} we consider the polynomial 
 $$h(x) = \det(VXB^\top),$$
where $X$ is a diagonal matrix with $X_{i,i}:=x_i$, $B\in \R^{d \times m}$ is the linear representation of $\cB$ and $V\in \R^{d\times {m}}$ is such that $V^\top V=L$. 
We remark that this polynomial has also appeared in  previous work on {\em matroid intersection} and {\em matroid parity}, e.g., in \cite{Lovasz79,Harvey09,CheungLL14}. We observe that $|h(x)|$ is $2$-anti-concentrated and has a number of desirable properties, which allows us to prove 

\begin{theorem}[Subdeterminant maximization under regular matroid constraints]\label{thm:main1} 
There exists a polynomial-time randomized algorithm such that given a PSD matrix $L\in \R^{m\times m}$ of rank $d$, and a totally unimodular matrix $B$ that is  a representation of a rank-$d$ regular matroid with  bases $\cB \subseteq 2^{[m]}$, outputs a set $S\in \mathcal{B}$ such that with high probability
$$\det(L_{S,S}) \geq \max(2^{-O(m)}, 2^{-O(d \log m)})\cdot \OPT,$$
  where $\OPT:= \max_{S\in \mathcal{B}}\det(L_{S,S})$.
\end{theorem}
There are two recent results for this setting (\cite{SV17} and \cite{AO17}) that provide  $e^m$- and $e^d$-estimation algorithms respectively. 
As in the case of the algorithm for partition matroids, these results  only give an estimate on the value of the optimal solution, and are not constructive.
Our algorithm   matches the approximation guarantee of the above-mentioned results in certain regimes and also  outputs an approximately optimal set.

\subsection{Discussion and future work}

To summarize, motivated by applications in machine learning, we propose and analyze two algorithms for subdeterminant maximization under matroid constraints. 
Both  are based on random sampling and the bounds on their approximation guarantees follow from our  anti-concentration result. 
These algorithms provide both an estimate of the value of the optimal solution as well as a set with the claimed guarantee. 
The anti-concentration inequality allows us to relate the value of a multi-variate nonconvex function at a random point to its value at the optimal point, and multi-linearity allows us to round this random solution.
Furthermore, the anti-concentration result can be applied to \emph{any} multi-linear polynomial and beyond.
In particular, it neither relies on real stability nor any other convexity-like property of the polynomial; this should be of independent interest.
We leave open the problem of extending the anti-concentration inequality from hypercubes and products of simplices to more general bodies; this might allow us to improve the approximation ratios.

\subsection{Other related work} 
A very general anti-concentration result for polynomial functions over convex domains was obtained by~\cite{CW}. However, there seem to be two issues in applying their result to our setting: 1) when specialized to our setting, it implies a weaker bound of $r^{-O(r)}$ in Equation ~\eqref{eq:anti_poly} to obtain a significant probability of success and, 2) it does not seem to directly apply to the kind of domains we consider in this paper (cartesian products of simplices).
A more detailed discussion is presented in Section~\ref{sec:thm13}.
The above-mentioned result by~\cite{CW} and, more generally, the anti-concentration phenomena has found several applications in theoretical computer science, especially for Gaussian measures; see for instance \cite{O'Donnell2014,DeDS16,Costello06, RV13}.
%

\section{Technical Overview}
 We start by describing the approach of~\cite{NS16} for the case of partition matroids.
Consider the following simple variant of the constrained subdeterminant maximization problem for partition matroids: Given vectors  $v_1,  \ldots, v_r , u_1, \ldots, u_r \in \R^r$ the goal is to pick a vector $w_i \in \{v_i, u_i\}$ for each $i$ so as to maximize $\abs{\det(W)}$, where $W \in \R^{r\times r}$ is a matrix that has the $w_i$s as its columns. 
Denote by $\OPT$ the maximum  value of the determinant in the above problem. 

They  start by reformulating the problem as  polynomial maximization problem as follows.
First, define matrices  $A_i(x_i) := x_i v_i v_i^\top + (1-x_i) u_i u_i^\top$ for $i=1,2, \ldots, r$. 
Then,  consider the polynomial $p(x,y) := \det \left( \sum_{i=1}^r y_i A_i(x_i)\right)$ and let $g(x)$ be the polynomial that appears as the coefficient of $\prod_{i=1}^r y_i$ in $p(x,y).$\footnote{$g(x)$ is also called the mixed-discriminant of the matrices $A_i(x_i)$.}
Multi-linearity of $g$  can be used to reduce the task of finding $\OPT$ to that of finding $\max_{x\in [0,1]^r} g(x)$.
Then, the  difficulty that arises  is that $g(x)$ is hard to evaluate. 
To bypass this,  a general idea by \cite{Gurvits06}  allows them to approximate $g(x)$ by  
$ \inf_{y>0} \frac{p(x,y)}{\prod_{i=1}^r y_i}$,
giving rise to the following  optimization problem involving two sets of variables
\begin{equation}\label{eq:ns}
{ \max_{x\in [0,1]^r} \inf_{y>0} \frac{p(x,y)}{\prod_{i=1}^r y_i}.}
 \end{equation}
Real stability of $p(x,y)$ for any fixed $x$ implies that  this program can be efficiently solved  using  convex programming. 
Their main result is that the value of this program is within a factor of $e^r$ of $\OPT$.
The key component in the proof of this bound is the above-mentioned result~\cite{Gurvits06} that, in this context where $p(x,y)$ is real-stable with respect to $y$, implies that, for all $x \in [0,1]^r$ 
\begin{equation}\label{eq:gurvits} { g(x) \leq  \inf_{y>0} \frac{p(x,y)}{\prod_{i=1}^r y_i} \leq e^r \cdot g(x).}
\end{equation}
While this immediately implies that one can obtain a number that is within an $e^r$ factor of $\OPT$, when trying to obtain an integral solution $x \in \{0,1\}^r$ from the fractional optimal solution $x^\star \in [0,1]^r$ to~\eqref{eq:ns}, the intractability of $g(x)$ becomes a bottleneck.\footnote{One can use Equation \eqref{eq:gurvits} $r$ times to give  an approximation algorithm with factor $e^{r^2}$; we omit the details.}
The authors of \cite{NS16} present a rounding algorithm which, unfortunately, can require an exponential number of trials to find an $e^r$-approximate solution.

\noindent 
{\bf Overview of the proof of Theorem \ref{thm:main2}.} Our approach is based on a different formulation of the problem as polynomial maximization, which has the advantage over $g(x)$ that {\em it is easy to evaluate and does not rely on real-stability}. 
For every $i=1,2, \ldots, r$ and $t\in [0,1]$ define a vector $w_i(t):=(1-t)v_i + tu_i$. 
Furthermore, for $x\in [0,1]^r$, let $W(x) \in \R^r$ be a matrix with columns $w_1(x_1), w_2(x_2), \ldots, w_r(x_r)$.
The polynomial that we consider is 
$$ {\det(W(x))}$$
which is easy to evaluate for any $x$.
As before, the multi-linearity of $\det(W(x))$ implies the following:
\begin{equation}\label{eq:our}
\max_{x\in [0,1]^r} |\det(W(x))| = \max_{x\in \{0,1\}^r} |\det(W(x))|=\OPT.
\end{equation}
Indeed, if we let $f(x):=|\det(W(x)|$, then
 the multi-linearity of $\det(W(x))$  implies that whenever we fix all but one of the arguments of $f$, i.e.,  $s(t):=f(t, y_2,y_3, \ldots, y_r)$ for some $y_2, y_3, \ldots, y_r \in [0,1]$, then $s$ attains its maximum at either  $0$ or  $1$. 
This means, in particular, that given any point $x\in [0,1]^r$, one can {efficiently find a point} $\tilde{x} \in \{0,1\}^r$ such that $f(\tilde{x}) \geq f(x)$.

However, the nonconvexity of  this formulation is a serious obstacle to solving the optimization problem  in Equation \eqref{eq:our}.
This is where a key insight comes in:  $f$ shows a remarkable anti-concentration property which, in turn, allows us to get an estimate of $\OPT$ by evaluating $f$ at a random point.
Formally, the anti-concentration inequality (Theorem \ref{thm:main0}) applies to $f$ and allows us to deduce that 
$$\textstyle{\pr \insquare{{f(x)} \geq c^{-r}\cdot \OPT }\geq \frac{1}{2}}$$
for some constant $c>1$.
This also results in a simple approximation algorithm to maximize $f$: Sample a point $x\in [0,1]^r$ uniformly at random, round $x$ to a vertex $\tilde{x}\in \{0,1\}^r$ such that $f( \tilde{x}) \geq f(x)$ as above,  and output $\tilde{x}$ as a solution. 
We should mention that at this point we could also attempt to invoke the  following anti-concentration result (here translated to our setting) proved by~\cite{CW}.
\begin{theorem}[Theorem 2 in~\cite{CW}] \label{thm:cw}
Let $p\in \R[x_1, x_2, \ldots, x_r]$ be a polynomial of degree $r$. If a point $x$ is sampled uniformly at random from the hypercube $[0,1]^r$, then for every $\beta\in (0,1)$ 
$$\textstyle{\pr\insquare{\abs{p(x)}\leq \beta^r\cdot  \OPT} \leq C\cdot  \beta\cdot r,}$$ 
where $C>0$ is an absolute constant. 
\end{theorem}
When applied to our setting, observe that $\det(W(x))$ is indeed a degree-$r$ polynomial in $r$ variables. 
We have to pick $\beta$ so as to make $C\cdot \beta \cdot r<1$, i.e., for $\beta = O(\nfrac{1}{Cr})$, we obtain
$$\textstyle{\pr\insquare{f(x) \geq r^{-O(r)} \cdot \OPT} \geq \frac{1}{2}.}$$
This implies that the algorithm described above achieves an approximation ratio of (roughly) $r^r$.
Our Theorem~\ref{thm:main0} is a certain strengthening of Theorem~\ref{thm:cw} which asserts that under the same assumptions
$$\pr \insquare{\abs{p(x)} \geq c^{-r}\cdot \OPT }\geq \frac{1}{2},$$
for some absolute constant $c>1$. 
In fact, Theorem~\ref{thm:main0} is a generalization of the above for a larger class of functions (not only polynomials)  and for more general domains -- this is useful in the case of general partition matroids.

{We now show how to extend our algorithm to a general  instance of the constrained  subdeterminant maximization problem under partition constraints and sketch a proof of Theorem~\ref{thm:main2}. 
Recall that in this problem we are given a PSD matrix $L\in \R^{m \times m}$ of rank $d$ and a partition family $\cB$ induced by a partition of $[m]$ into disjoint sets $M_1, M_2, \ldots, M_t$ and numbers $b_1, b_2, \ldots, b_t \in \N$ with $\sum_{i=1}^t b_i = r$. 
The goal is to find a subset $S\in \cB(\cM)$ such that $\det(L_{S,S})$ is maximized. 
If we consider a decomposition of $L$ into $L=V^\top V$ for $V\in \R^{d \times m}$ then the objective can be rewritten as $\det(L_{S,S}) = \det(V_S^\top V_S)$. 
For simplicity, we  assume that $b_1=b_2= \cdots = b_t=1$, which can be achieved by a simple reduction. 
To define the relaxation for the general case, for every part $M_i$ for $i=1,2, \ldots, t$, introduce a vector $x^i \in \Delta_{p_i}$ where $p_i := |M_i|$ and define a vector $w^i(x^i)$ to be
$$w^i(x^i) := \sum_{j=1}^{p_i} x_j^i v_j^i$$
where $v_1^i, v_2^i, \ldots, v_{p_i}^i$ are the columns of $V$ corresponding to indices in $M_i$. 
We denote by $x$ the vector $(x^1, x^2, \ldots, x^r)$ and by $W(x) \in \R^{d\times r}$ the matrix with columns $w^1(x^1), w^2(x^2), \ldots, w^r(x^r)$. 
Finally we let 
$$f(x^1, x^2, \ldots, x^r) := \det(W(x)^\top W(x))^{1/2}.$$
Note that $f(x)$ is no longer a multi-linear polynomial, but as we  show in Lemma~\ref{lemma:anti_conc_norm} it is $2$-anti-concentrated. 
Having established this property,  Theorem~\ref{thm:main2} follows. 
Indeed, as in the illustrative example in the beginning, we can prove that given any fractional point $x$, we can efficiently find its integral rounding (i.e., round every component $x^i$ to a vertex of the corresponding simplex $\Delta_{p_i}$, for $i=1,2, \ldots, t$) which then provides us with a suitable approximate solution.}

\paragraph{Overview of the proof of Theorem \ref{thm:main1}.} 
In the setting of Theorem~\ref{thm:main1} we are given a PSD matrix $L\in \R^{m\times m}$ of rank $d$ and a family of bases $\cB \subseteq 2^{[m]}$ of a regular matroid of rank $d$. 
The goal is to find a set that attains $\OPT:=\max_{S\in \cB} \det(L_{S,S}).$
The approach of  \cite{SV17} to obtain an estimate on $\OPT$ was inspired by that of \cite{NS16} for the partition matroid case and is as follows:\footnote{The approach of \cite{AO17} is also similar.} 
Given the matrix $L=V^\top V$, first,  define the following polynomial 
$$\textstyle{g(x) := \sum_{S \in \cB} x^S \det (V_S^\top V_S).}$$
This polynomial again turns out to be hard to evaluate. 
As before,  an optimization problem involving two sets of variables, $x$ and $y$ is set up.
The purpose of $y$ variables is to give  estimates of values of  $g(x)$ and the $x$ variables are constrained to be in the matroid base polytope corresponding to $\cB$. 
On the one hand, real stability along with the fact that $\cB$ is a matroid  allows them to compute the optimal solution to this bivariate problem,  
on the other hand, with some additional effort, they are able to push the result of~\cite{Gurvits06} to obtain roughly an $e^m$ estimate of $\OPT$.
However,  the main bottleneck is that an iterative rounding approach for finding an approximate integral point does not seem possible as the matroid polytope corresponding to $\cB$ may not have a product structure as in the partition matroid case. 

We present a new formulation to capture $\OPT$ that does not suffer from the intractability of the objective function and allows for rounding via a relaxation that maximizes a certain function $h$ over the hypercube $[0,1]^m$. 
Start by noting that the objective becomes $\det(L_{S,S}) = \det(V_S^\top V_S) = \det(V_S)^2$, which we can simply think of as maximizing $\abs{\det(V_S)}$ over $S\in \cB$. 
Let $B\in \Z^{m\times d}$ be the linear representation of the matroid $\cB$; i.e., for every set $S\subseteq [m]$ of size $d$, if $S\in \cB$ then $|\det(B_S)|=1$, and $\det(B_S)=0$ otherwise. Next, consider $h: [0,1]^m \to \R$ given by
$$h(x) := \det(VXB^\top),$$
where $X\in \R^{m\times m}$ is a diagonal matrix with $X_{i,i}:=x_i$ for all $i=1,2,\ldots, m$. 
It is not hard to see that $h(x)$ is a polynomial in $x$ and (using the Cauchy-Binet formula) can be written as
$${ h(x) =\sum_{S\subseteq[m], |S|=d} x^S \det(V_S) \det(B_S),}$$
where $x^S$ denotes $\prod_{i \in S} x_i$. 
Such a function was studied before in the context of matroid intersection problems \cite{Lovasz89,Harvey09,GT17}.  
Importantly, the restriction of $h(x)$ to indicator vectors of sets of size $d$ is particularly easy to understand. Indeed, let $1_S$ be the indicator vector of some set $S\subseteq [m]$ with $|S|=d$. 
We have
$$\textstyle{h(1_S) = \det(V_S) \det(B_S) = \begin{cases}\pm \det(V_S) &\mbox{if }S\in \cB, \\
0 & \mbox{if } S\notin \cB. \end{cases}}$$
Hence, we are interested in the largest magnitude coefficient of a multi-linear polynomial $h(x)$. The maximum of $|h(x)|$ over $[0,1]^m$ is an upper bound for this quantity. 
The algorithm then simply selects a point $x\in [0,1]^m$ at random, which by Theorem~\ref{thm:main0} can be related to the maximum value of $|h(x)|$, and then performs a rounding.

First, given $x\in [0,1]^m$ it constructs a binary vector $\tilde{x} \in \{0,1\}^m$ such that $|h(\tilde{x})| \geq |h(x)|$; this is possible because the function $|h(x)|$ is  convex along any coordinate direction. 
The vector $\tilde{x}$ is then treated as a set $S_0 \subseteq [m]$, but its cardinality is typically  larger than $d$. 
We then run another procedure which repeatedly removes elements from $S_0$ while not loosing too much in terms of the objective. It is based on using $h(1_{S_0})$ as a certain proxy for the sum $\sum_{S\subseteq S_0} |\det(V_S) \det(B_S)|$.  This allows us to finally arrive at a set $S\subseteq S_0$ of cardinality $d$, such that $\abs{h(1_S)} \geq {m \choose d}^{-1} \abs{h(1_{S_0})}$. The set $S$ is then the final output.
By applying Theorem~\ref{thm:main0} one can conclude that $h(1_{S_0})$ is within a factor of $c^m$ of the maximal value of $|h(x)|$, which results in a $2^{O(m)}$-approximation guarantee for the algorithm. 
Alternatively, by utilizing the fact that $h$ is a polynomial of degree $d$, one can apply the result by Carbery-Wright (see Theorem~\ref{thm:cw}) to obtain a bound of roughly $m^{O(d)}$, which is better whenever $m$ is large compared to $d$.

\paragraph{Overview of the proof of Theorem \ref{thm:main0}.} 
For the sake of clarity, we present only  the hypercube case of the anti-concentration inequality, which corresponds to taking $p_1=p_2=\cdots = p_r=2$ in the statement of Theorem~\ref{thm:main0}. 
Recall the setting: We are given a function $f:[0,1]^r \to \R_{\geq 0}$ that satisfies a one-dimensional anti-concentration inequality. I.e., for every function of the form $g(t):= f(x_1, x_2, \ldots, x_{i-1}, t, x_{i+1}, \ldots, x_r)$ where $x_j\in [0,1]$ for $j\neq i$ are fixed and $t\in [0,1]$, it holds that
\begin{equation}\label{eq:1dim}
\pr\insquare{g(t)< c \cdot \max_{t\in [0,1]} g(t)}\leq 2c,
\end{equation}
where the probability is over a random choice of $t\in [0,1]$. 
The goal is to prove a similar statement for $f(x)$, i.e., $\pr\insquare{f(x)< \alpha \cdot  \OPT}$ is small, where $\OPT$ is the maximum value $f$ takes on the hypercube and $\alpha$ is a parameter which we want to be as large as possible. 

As an initial approach, one can define (for some $c>0$) events of the form
$$A_i :=\inbraces{x\in [0,1]^r: f(x_1, x_2, \ldots, x_i, x_{i+1}^\star, \ldots, x_r^\star) \geq c \cdot f(x_1, x_2, \ldots, x_{i-1}, x_i^\star, \ldots, x_r^\star)},$$
where $x^\star := \argmax_x f(x)$. 
Note crucially that the events $A_1, A_2, \ldots, A_r$ are not independent.
However, we can still write
\begin{align*}
\pr\insquare{f(x) \geq c^r \cdot \OPT} & \geq \pr\insquare{ A_1 \cap A_2 \cap A_3 \cdots \cap A_r}\\
& = \pr[A_1] \cdot \pr[A_2 |A_1] \cdot \pr[A_3| A_1, A_2]  \cdots  \pr[A_r| A_1, A_2, \ldots, A_{r-1}].
\end{align*}
From  assumption~\eqref{eq:1dim} we know that
$$\pr[A_i|A_1,A_2, \ldots, A_{i-1}] \geq 1-2c$$
for all $i=1,2,\ldots, r$ and hence 
$$\pr\insquare{f(x) \geq c^r \cdot \OPT} \geq (1-2c)^r.$$
To get a probability that is not exponentially small, one has to take the value of $c$ roughly $O(\nfrac{1}{r})$, in which case we recover  the result by Carbery and Wright~\cite{CW} in our setting.
To go beyond this, a tighter analysis is required. 
For this we consider the random variables 
$$Z_i:= \frac{f(x_1, x_2, \ldots, x_{i}, x_{i+1}^\star, \ldots, x_r^\star)}{f(x_1, x_2, \ldots, x_{i-1}, x_i^\star, \ldots, x_r^\star)}$$
for $i=1,2,\ldots, r$. Note that $\prod_{i=1}^ r Z_i = \frac{f(x)}{\OPT}$ hence the goal reduces to proving that
$$\pr\insquare{\prod_{i=1}^ r Z_i \geq c^{-r}}\geq \frac{1}{2}.$$
To obtain such a bound, we first translate the product to a sum and define $X_i:= - \log Z_i$ which then reduces our task to proving
\begin{equation}\label{eq:tail_prob}
\pr \insquare{\sum_{i=1}^r X_i \leq O(r)}\geq \frac{1}{2}.
\end{equation}
the anti-concentration assumption on $f$ translates to the following convenient bound on the CDF of $X_i$
$$\pr[X_i \geq t|X_1, X_2, \ldots, X_{i-1}] \leq \min(1, 2e^{-t})~~~~~\forall_{t\in \R}.$$
Hewever, again, the fact that $X_i$'s are not independet presents itself as a hurdle.
To overcome it, we prove a monotonicity result (Lemma~\ref{lemma:prob_monotone}). It asserts that one can replace the variables $X_i$ in~\eqref{eq:tail_prob} by independent copies $Y_i$ of a random variable with CDF $t\mapsto \min(1, 2e^{-t})$. 
After establishing this fact, it remains to obtain a tail bound for independent variables, for this, we simply apply Chebyshev's inequality.

 \subsection{Organization of the rest of the paper}
We present some notation and preliminaries about matroids and measures in Section \ref{sec:preliminaries}.
In Section \ref{sec:thm11} we present the proof of our anti-concentration result, Theorem \ref{thm:main0}. In Section \ref{sec:thm13} we present a proof of Theorem \ref{thm:main2} for partition matroids. 
In Section \ref{sec:thm12} we present a proof of Theorem \ref{thm:main1} for regular matroids.
In Section \ref{sec:nshard} we present an example to show that the~\cite{NS16} algorithm may not yield a polynomial-time approximation algorithm.

\section{Preliminaries}\label{sec:preliminaries}
\parag{Simplices and Measures}  The $d$-dimensional Lebesgue measure (volume) on $\mathbb{R}^d$ is denoted by $\lambda_d$. 
When the dimension is clear from the context, we use $\lambda$ to denote the volume.  
Throughout this chapter, the probability distributions we consider, are typically uniform over an appropriate domain. 

The standard $(d-1)$-simplex, denoted by $\Delta_d$ is defined as the convex hull of $e_1, e_2, \ldots, e_d \in \R^d$. 
Notice that $\Delta_d$ is a $(d-1)$-dimensional polytope which is embedded in $\mathbb{R}^d$, and it inherits a $(d-1)$-dimensional Lebesgue measure from the hyperplane it lies on. 
We use $\mu_d$ to denote the induced measure $\lambda_d$ on the simplex $\Delta_d$, normalized so that $\mu_d(\Delta_d)=1$. 
We often deal with Cartesian products of simplices, which we denote by $\Delta = \prod_{i=1}^r \Delta_{p_i}$, for some sequence $p_1, p_2, \ldots, p_r \in \N$. 
For a point $x\in \Delta$, by $x^{i}$ we denote $i$-th component of $x$ belonging to $\Delta_{p_i}$ and $x^{i}_j$ for $j\in [p_i]$ are the components of $x^i$ within $\Delta_{p_i}$. 
By $V(\Delta)$, we denote the set of points of $\Delta$ with integer coordinates. We call $V(\Delta)$ the set of vertices of $\Delta$.

 \vspace{2mm}
 \parag{Multi-linear functions. } 
 A function $f:\R^m \to \R$ is called multi-linear if $f$ is a polynomial function where the degree of each variable is at most 1. 
Suppose that $x_1,\ldots,x_m$ are $m$ variables. We denote the monomial $\prod_{i\in S}{x_i}$ by $x^S$ for every $S\subseteq [m]$.  Every multi-linear function can be written in the form $f(x)=\sum_{S\subseteq [m]}{f_Sx^S}$ where $f_S$'s are real numbers, called the coefficients of $f$. 
A function $f:\R^m \to \R$ is called affine when $f$ is a polynomial whose total degree is at most one. 
A function $f: \mathbb{R}^{p_1}\times \cdots \times \mathbb{R}^{p_r}\to \mathbb{R}$ is called block-multi-linear if for every index $i\in [r]$ and for every choice of $y^j\in \mathbb{R}^{p_j}, j \in [r]\setminus \{i\}$ the function $f(y^1,\ldots,x^i,\ldots,y^r)$ is an affine function over $\mathbb{R}^{p_i}$.

\vspace{2mm}
 \parag{Matroids. } 
For a comprehensive treatment of matroid theory we refer the reader to~\cite{Oxley06}. 
Below we state the most important definitions and  examples of matroids, which are most relevant to our results. 
A matroid is a pair $\cM=(U, \cI)$ such that $U$ is a finite set and $\cI \subseteq 2^{U}$ satisfies the following three axioms: (1) $\emptyset \in \cI$, (2) if $S\in \cI$ and $S'\subseteq S$ then $S'\in \cI$, (3) if $A,B\in \cI$ and $|A|>|B|$, then there exists an element $a\in A\setminus B$ such that $B\cup \{a\} \in \cI$. 
The collection $\cB \subseteq \cI$ of all inclusion-wise maximal elements of $\cM$ is called the set of bases of the matroid. 
It is known that all the sets in $\cB$ have the same cardinality, which is called the rank of the matroid. 
In this paper we often work with sets of bases $\cB$ of matroids instead of independent sets $\cI$, for this reason we will also refer to a pair $(U,\cB)$ as a matroid.

\vspace{2mm}
\noindent \textbf{Linear and regular matroids.} 
Let $U=\{W_1,W_2,\ldots,W_m\}\subseteq \R^n$ be a set of vectors. 
Let $\mathcal{B}$ consists of all subsets of $U$ which form a basis for the linear space generated by all the vectors in $U$. 
$\cM=(U,\mathcal{B})$ is called a linear matroid.  
A matrix $A\in \R^{r\times m}$ is called a representation of a matroid $\mathcal{M}=([m],\mathcal{B})$, if for every set $S\subseteq [m]$, $S$ is independent in $\cM$ if and only if the corresponding set of columns $\{A_i: i\in S\}$ is linearly independent.  
A matroid $\cM=(M,\mathcal{B})$ is called a regular matroid if it is representable by a totally unimodular real matrix. 
A matrix is called totally unimodular if the determinant of any of its square submatrices  belongs to the set $\{-1,0,1\}$.  
 
 \vspace{2mm}
\noindent \textbf{Partition matroids.} 
A matroid $\cM=(M,\mathcal{B})$ is said to be a partition matroid if there exists a partition $\mathcal{P}=\{M_1,M_2,\ldots,M_t\}$ of the ground set $M$ and a sequence of non-negative integers  $b=(b_1,b_2,\ldots,b_t)$ such that $|B\cap M_i|=b_i$ for all $B\in \mathcal{B}$ and $i=1,2,\ldots,t$. 

\section{Anti-Concentration Inequality: Proof of Theorem \ref{thm:main0}}\label{sec:thm11}
Before starting the proof of Theorem~\ref{thm:main0} we first need to establish a certain monotonicity result that allows us to replace dependent random variables by their independent copies when deriving tail bounds.
\begin{lemma}[Monotonicity]\label{lemma:prob_monotone}
Let $Y_1, Y_2, \ldots, Y_r$ be real random variables with CDFs $f_1, f_2, \ldots, f_r: \R \to [0,1]$ respectively, i.e., $f_i(x):=\pr[Y_i\leq x]$ for $i\in [r]$ and $x\in \R$. Suppose $X_1, X_2, \ldots, X_r$ are real random variables such that
$$\pr[X_i \leq x|X_1, X_2, \ldots, X_{i-1}]\geq f_i(x)~~~~~~~\mbox{for every }i=1,2, \ldots, r \mbox{ and } x\in \R$$
then for every function $G:\R^r \to \R_{\geq 0}$ which is monotone with respect to every coordinate it holds
$$\expect{G(X_1, X_2, \ldots, X_r)} \leq \expect{G(Y_1, Y_2, \ldots, Y_r)}.$$
\end{lemma}
\begin{proof}
We will prove the claim by induction on $r$. Consider the case of $r=1$ first. Let $g_1$ be the CDF of $X_1$, we have

$$\expect{G(X_1)} = \int G(x_1) d g_1(x_1).$$
where the above is a Riemann-Stieltjes integral with respect to $g_1$. Since $G$ is monotone and $g_1 \geq f_1$, it is an elementary fact on R-S integrals that
$$\int G(x_1) d g_1(x_1)\leq G(x_1) d f_1(x_1)=\expect{G(Y_1)}$$
and hence the claim for $r=1$.

Suppose now that the claim holds for  $(r-1) \in \N$, we will prove it for $r$. Denote
\begin{align*}
H(X_1, X_2, \ldots, X_{r-1})&:=\expect{G(X_1, \ldots, X_r)| X_1, \ldots, X_{r-1}}\\
K(Y_1, Y_2, \ldots, Y_{r-1})&:=\expect{G(Y_1, \ldots, Y_r)| Y_1, \ldots, Y_{r-1}}
\end{align*}
(the conditional expectations). From the assumption and the $r=1$ case we have that for every tuple $(x_1, \ldots, x_{r-1})\in \R^{r-1}$ we have
$$H(x_1, \ldots, x_{r-1}) \leq K(x_1, \ldots, x_{r-1}).$$
Further:
\begin{align*}
\expect{G(X_1, \ldots, X_r)} &= \expect{H(X_1, X_2, \ldots, X_{r-1})}\\
&\leq  \expect{H(Y_1, Y_2, \ldots, Y_{r-1})}\\
&\leq  \expect{K(Y_1, Y_2, \ldots, Y_{r-1})}\\
&= \expect{G(Y_1, \ldots, Y_r)}
\end{align*}
where the transition from the first to the second line follows from the induction hypothesis. 
\end{proof}
\noindent 
Given the above lemma we are ready to prove Theorem~\ref{thm:main0}.
\begin{proofof}{of Theorem \ref{thm:main0}}
Let us fix any optimal point $x^\star:=(x_1^\star, \ldots, x_r^\star)$, i.e., such that $f(x_1^\star, \ldots, x_r^\star) = \OPT$ and consider random variables
$$Z_i:=p_i \frac{f(x_1, x_2, \ldots, x_{i}, x_{i+1}^\star, \ldots, x_r^\star)}{f(x_1, x_2, \ldots, x_{i-1}, x_i^\star, \ldots, x_r^\star)}$$
under a uniformly random choice of $x\in \prod_{i=1}^r \Delta_{p_i}$. We now have
$$\frac{f(x_1, \ldots, x_r)}{\OPT}\cdot \prod_{i=1}^r p_i  =\prod_{i=1}^r Z_i,$$
because $\prod_{i=1}^r Z_i$ is a telescoping product. Since $r\geq 2$ to prove the theorem it suffices to show that
$$\pr \insquare{\prod_{i=1}^r Z_i\leq (\gamma e^{2})^{-r}}\leq \frac{1}{r}.$$
From the definition of anti-concentration, for every $i\in [r]$ and every $c\in (0,1)$ we have
\begin{equation}\label{eq:cdf}
\pr[Z_i\leq c|Z_1, \ldots, Z_{i-1}]\leq \gamma c.
\end{equation}
In this case it is more convenient to analyze sums than products, hence let us define
$$X_i:= - \log Z_i ~~~~~~~\mbox{for all } i=1,2, \ldots, r.$$
Now our task reduces to finding an upper bound on the probability $\pr\insquare{\sum_{i=1}^r X_i \geq \Omega(r)}$. From~\eqref{eq:cdf} we obtain that for every $x\in \R_{\geq 0}$ we have

$$\pr[X_i\geq x|X_1, \ldots, X_{i-1}] \leq \gamma e^{-x}.$$
Let us now define $Y_1, \ldots, Y_r \in \R_{\geq 0}$ to be independent random variables such that for every $i\in [r]$ and $x\in \R_{\geq 0}$ 
\begin{equation}\label{eq:distr_y}
\pr[Y_i \geq x] = \min \inparen{1, \gamma e^{-x}}.
\end{equation}
We claim that Lemma~\ref{lemma:prob_monotone} implies that for every $x\in \R_{\geq 0}$ it holds that
$$\pr \insquare{\sum_{i=1}^r X_i \geq x} \leq \pr \insquare{\sum_{i=1}^r Y_i \geq x}.$$
Indeed, to arrive at such a conclusion one can consider the function
$$G(t_1, \ldots, t_r) := \insquare{\sum_{i=1}^r t_i \geq x}$$
where $\insquare{\phi}$ is the Iverson bracket, i.e., it is $1$ when $\phi$ holds and $0$ otherwise. The function $G$ is clearly monotone and 
$$\expect{G(X_1, X_2, \ldots, X_r)}=\pr \insquare{\sum_{i=1}^r X_i \geq x}.$$
It is now enough to derive a bound on $\pr \insquare{\sum_{i=1}^r Y_i \geq x}$ for independent variables $Y_1, \ldots, Y_r$ distributed as in~\eqref{eq:distr_y}. To this end we simply apply the Chebyshev's inequality. Let us now compute the expectation and variance of a variable $Y$ distributed as the $Y_i$'s.

Let us denote the density of $T$ by $g(x):=\gamma e^{-x}$, note also that $Y\in (\log \gamma, \infty)$. We have

$$\expect{Y} = \int_{\log \gamma}^\infty y (-g'(y))dy=1+\log(\gamma).$$ 
Similarly
$$\expect{(\expect{Y}-Y)^2} = \int_{\log \gamma}^\infty (y-1-\log(\gamma))^2 (-g'(y))dy=1.$$
\noindent 
Now from Chebyshev's inequality we obtain that for any $M>0$ 
$$\pr\insquare{\sum_{i=1} Y_i \geq \expect{\sum_{i=1}Y_i} + M} \leq \frac{\mathrm{Var}\insquare{\sum_{i=1}Y_i}}{M^2},$$
and hence 
$$\pr\insquare{\sum_{i=1} Y_i \geq r(1+\log \gamma) + M} \leq \frac{r}{M^2}.$$
Thus by taking $M:=r$ we obtain
$$\pr\insquare{\sum_{i=1} Y_i \geq r(2+\log \gamma) } \leq \frac{1}{r}.$$
Finally, translating this bound to $X_i$'s and then to $Z_i$'s we conclude
$$\pr \insquare{\prod_{i=1}^r Z_i\leq \inparen{\gamma e^2}^{-r}}\leq \frac{1}{r},$$
which concludes the proof.
\end{proofof}

\section{Partition Matroids: Proof of Theorem~\ref{thm:main2}}\label{sec:thm13}

\subsection{Lemma on Anti-concentration} 
We start by showing the following lemma saying that the $\ell_2$ norm of a convex combination of vectors is anti-concentrated.

\begin{lemma}\label{lemma:anti_conc_norm}
Let $w_1, w_2, \ldots, w_p \in \R^d$ be any vectors. Then the function $f: \Delta_p \to \R$ defined as $f(x)=\norm{\sum_{j=1}^p x_j w_j}$ is $2$-anticoncentrated.
\end{lemma}

\begin{proof}
We begin by establishing the fact for $p=2$. To this end define $g(x):=\abs{x_1 \norm{w_1} - x_2 \norm{w_2}}$, we claim that 
$$\forall_{x\in \Delta_2}~~~ g(x) \leq f(x).$$
The above claim follows simply from triangle inequality. Indeed
$$\norm{x_1 w_1} = \norm{x_1 w_1 +x_2w_2 - x_2 w_2} \leq \norm{x_1w_1 +x_2 w_2} +\norm{x_2 w_2}$$
and hence
$$\norm{x_1 w_1} - \norm{x_2 w_2} \leq \norm{x_1w_1 +x_2 w_2}.$$
By symmetry $\norm{x_2 w_2} - \norm{x_1 w_1} \leq \norm{x_1w_1 +x_2 w_2}$ follows as well. 
Given the claim and observing that $\max_{x\in \Delta_2} f(x) = \max_{x\in \Delta_2} g(x)$, it is enough to prove $2$-anti-concentration of $g$, since then an analogous result for $f$ follows. This is in fact the subject of Fact~\ref{fact:simple_anti_conc}, hence the $p=2$ case follows.

The case of $p\geq 3$ is proved differently, by taking advantage of the $p=2$ case. The challenge to prove it comes from the fact that generating a random point from a high-dimensional simplex $\Delta_p$ is not equivalent to simply generating its coordinates independently and uniformly at random and then normalizing the obtained point so that it sums up to one. There are several known methods for sampling a random point from $\Delta_p$, however, no ``practical'' method seems to be well suited for this proof and below we simply use the basic definition to deal with it.

Consider any isometric embedding $P$ of $\Delta_p$ in $\R^{p-1}$ where it is a full-dimensional polytope. Then consider the uniform distribution over any box containing $P$. Conditioned on the sample landing in $P$, the corresponding distribution is -- by definition -- uniform on $P$ and thus (via the embedding) uniform on $\Delta_p$. 

 Denote the vertices of $\Delta_p$ in the embedding to be $v_1, v_2, \ldots, v_p \in \R^{p-1}$. Let also $\wt{g}: P \to \R$ be the corresponding function $g$ on $P$, i.e., 
 $$\forall_{x\in \Delta_p}~~\wt{g}\inparen{\sum_{j=1}^p x_j v_j}=g(x).$$
 Assume without loss of generality that $\norm{w_1}$ is the largest among $\norm{w_1}, \norm{w_2}, \ldots, \norm{w_p}$ and that $v_1=0$. Now, consider any point $v \in P$ on the facet opposite to $v_1$, i.e. $v=\sum_{j=2}^p y_j v_j$ where $(y_2, y_3, \ldots, y_p) \in \Delta_{p-1}$. For $z\in [0,1]$ consider
 $$h(z)=\wt{g}(zv_1 + (1-z) v)=\norm{zw_1 + (1-z) \sum_{j=2}^p y_j w_j}.$$
 From the $p=2$ case $h$ is $1$-anti-concentrated and moreover $\max_{z\in [0,1]} h(z) = \max_{x\in \Delta_p} f(z) = \norm{w_1}$. Thus for every ray $[v_1, v]$ ($v \in \conv\{v_2, \ldots, v_p\}$) we have an anti-concentrated function on it, whose maximum coincides with the maximum of $f$, and the simplex $P$ is a disjoint union of such rays. Seemingly, this already implies anti-concentration of $f$, however, note that the distribution on the ray $[v_1, v]$ induced from the uniform distribution over $P$ is not uniform and hence the result does not follow yet.
  
More formally, let us denote the distribution on $[0,1]$ which is induced from the uniform distribution on $P$ when restricted to $[v_1, v] \equiv [0,1]$ by $\mu_v$ we would like to prove:
$$\pr_{z\sim \mu_v}[h(z)< c \cdot \norm{w_1}] \leq 2pc$$
but what we know is only that when $z$ is uniformly distributed over $[0,1]$:
$$\pr_{z\sim \mathcal{U}[0,1]}[h(z)< c \cdot \norm{w_1}] \leq 2c.$$
Thus it remains to understand $\mu_v$. The density of $\mu_v$ can be derived from the hyperspherical coordinate system and its Jacobian. In fact it follows that for any fixed ray $[0,v]$ the density $\mu_v$ on $[0,1]$ at a point $z$ is proportional to $z^{p-2}$. Thus the task of proving anti-concentration finally reduces to the following inequality. Given a set $A\subseteq [0,1]$ of Lebesgue measure at most $2c$, show that
$$\int_A \frac{z^{p-2}}{p-1} dz \leq 2pc.$$
Because of monotonicity this is equivalent to proving (note that we may assume that $2c<1$ here)
$$\int_{1-2c}^1 \frac{z^{p-2}}{p-1} dz \leq 2pc,$$
which further reduces to 
$$1 - (1-2c)^{p-1} \leq 2pc$$
the above holds by Bernoulli's inequality.
\end{proof}

\begin{fact}\label{fact:simple_anti_conc}
Let $a_1, a_2 \in \R$ be any numbers. Consider the function $f: \Delta_2 \to \R$ given by $f(x)=\abs{a_1 x_1 +a_2 x_2}$. Then $f$ is $1$-anti-concentrated.
\end{fact}
\begin{proof}
Let us translate the question to a $1-$dimensional problem first. Let $a=\max(|a_1|, |a_2|)=\max_{x\in \Delta_2} f(x)$ and define $g:[0,1] \to \R$ by $g(t)=\abs{(1-t)a_1 + t a_2}$. We would like to prove that when $t$ is sampled uniformly at random from $[0,1]$ then for every $c\in (0,1)$ we have
$$\pr[g(t)<c\cdot a] \leq  2c.$$
Assume without loss of generality that $g(0)=a\geq g(1)$ and that $g$ is not a constant function.
There are two cases: either $g$ has a single root in $[0,1]$ or it has no roots. We analyze the former, as the latter then also follows.

Let $t_0 \in (0,1]$ be the root of $g(t)$. It is not hard to see that $t_0 \geq \frac{1}{2}$, as $g(t-t_0)$ is a symmetric function.  
Now, the function $g$ on $[0,t_0]$ is linear and hence
$$\pr_{t\in [0,t_0]}[g(t)<c\cdot a] \leq c,$$
and consequently 
$$\pr\insquare{(g(t)<c\cdot a) \wedge t\in [0,t_0]} \leq c\cdot t_0.$$
By symmetry and the fact that $g(1)\leq g(0)$ we have
$$\pr\insquare{(g(t)<c\cdot a) \wedge t\in [t_0,1]} \leq \pr\insquare{(g(t)<c\cdot a) \wedge t\in [0,t_0]},$$
and hence
$$\pr\insquare{g(t)<c\cdot a} \leq 2\cdot c\cdot t_0 \leq 2\cdot c.$$
\end{proof}

\subsection{Proof of Theorem \ref{thm:main2}}

\begin{proofof}{of Theorem \ref{thm:main2}}
We start by observing that it suffices to prove the Theorem for the case when  $b_1=b_2=\cdots=b_t=1$. 
Indeed, when $b_i$'s are not all equal to 1, we can perform a simple reduction to the all-ones case. Namely, we construct a new instance of the problem, where every part $M_i$ is repeated $b_i$ times. After doing so, we obtain a new instance with $r$ parts $M_1', M_2', \ldots, M_r'$ and $b_1'=b_2' = \ldots = b_r'=1$. 

Every feasible solution to the original instance corresponds to a feasible solution to the new instance (with the same value). Conversely, every feasible solution {\it with non-zero value} corresponds to a feasible solution in the original instance.

Finally, the bound on the approximation ratio follows easily by translating the bound in the simple case $b_1=b_2= \ldots = b_r=1$ to the instance after reduction.

From now on we assume that $b_1=b_2=\cdots=b_t=1$; in this case $t=r$. 
Let
\begin{equation*}
L=V^\top V
\end{equation*}
be the Cholesky decomposition of the PSD matrix $L$ with $V\in \R^{d\times m}$. 
One can easily see that
$$  
L_{S,S}=V_S^\top V_S,~~~\mbox{ for all } S\subseteq [m].
$$
For every part $M_i$ ($i=1,2, \ldots, t$) consider the $p_i$-simplex $\Delta_{M_i}$ indexed by the elements in $M_i$, i.e.
$$\Delta_{M_i}=\inbraces{y\in [0,1]^{M_i}: \sum_{j\in M_i} y_j=1}.$$
Further, consider $\Delta:=\prod_{i=1}^t \Delta_{M_i}$ and a function $f: \Delta \to \R$ defined as follows
$$ f(x):=\det \insquare{V(x)^\top V(x)}^{1/2}$$
where $V(x) \in \R^{d \times t}$ matrix, whose $i$th column is $V_i(x):=\sum_{j\in M_i} x_j v_j $.
\noindent
Note that when $x\in \Delta$ is a $0-1$ vector, i.e., $x=1_S$ for some set $S\in \cB$ then $f(x)^2=\det(V_S^\top V_S)$.
Thus, there exists a natural bijection between the elements of $ \cB$ (bases of the partition matroid) and the vertices of $\Delta=\prod_{i=1}^t \Delta_{p_i}$.  
Therefore, the optimization problem can be stated as the problem of maximizing $f$ over the vertices of $\Delta$. 
That is
\begin{equation}\label{eq:volumeoptimizationformulation}
\max ~\{ f(x): x\in \Delta \cap \{0,1\}^m \}.\\
\end{equation}
\noindent
We prove that maximizing $f$ over integer points in $\Delta$ is the same as maximizing it over the whole polytope $\Delta$.
This, composed with an algorithm to round a fractional point to a vertex and an anti-concentration result on $f$ will allow us to conclude Theorem~\ref{thm:main2}.
We start with the former.
Let us fix all but the first block-coordinates of $x\in \Delta$, i.e. $x=(y,x')$, where $y\in \Delta_{M_1}$ and $x'\in \prod_{i=2}^t \Delta_{M_i}$ is fixed.
Further, denote by $V'(x)$ the submatrix of $V(x)$ composed of columns $V_2(x), \ldots, V_t(x)$. By the formula on the determinant of a block matrix, we have
$$\det\insquare{V(x)^\top V(x)}=\det\insquare{V'(x)^\top V'(x)}\cdot\inparen{V_1(x)^\top \cdot \Pi \cdot V_1(x)}$$
where $\Pi\in \R^{(t-1)\times (t-1)}$ is a certain projection matrix. Thus in particular, there exist vectors $\{w_j\}_{s\in M_1}$ such that
\begin{equation}\label{eq:one_coord_fun}
f(x)=\norm{\sum_{j\in M_1} y_j w_j} \cdot \det\insquare{V'(x)^\top V'(x)}^{1/2}.
\end{equation}
Note that the above, as a function of $y\in \Delta_{M_1}$ is maximized at some vertex $y\in \Delta_{M_1} \cap \{0,1\}^{M_1}$. 
And thus (by induction), the whole function $f(x)$ is maximized at an integer vector.
This observation also implies a simple rounding algorithm: given any fractional point $x\in \Delta$, go coordinate by coordinate $i=1,2, \ldots, t$ and round it to a vertex which provides the largest value of $f$, this requires to evaluate $f$ at $p_i$ points only. 

Thus so far we have proved that an (approximation) algorithm for finding a fractional maximizer of $f$ over $\Delta$ can be turned into an algorithm maximizing $\det(V_S^\top V_S)$ over $S\in \cB$ with the same guarantee and polynomial overhead in the running time. 

We prove that $f$ is $2$-anticoncentrated which implies that a value of $f$ at a random point gives, with high probability, a decent estimate of the optimal value. 
In fact, $2$-anticoncentration, together with Theorem~\ref{thm:main0} and the observation above implies Theorem~\ref{thm:main2} immediately.

To prove anticoncentration, we need to analyze how does $f$ behave when all but one of its coordinates are fixed.
Without loss of generality fix all but the first coordinate.
Note that by~\eqref{eq:one_coord_fun} our goal becomes to prove that the function $\Delta_{M_1} \ni y\mapsto \norm{\sum_{j\in M_1} y_j w_j}$ is $2$-anticoncentrated.
However, this exactly what we prove in Lemma~\ref{lemma:anti_conc_norm}.
\end{proofof}

\section{Regular Matroids: Proof of Theorem~\ref{thm:main1}}\label{sec:thm12}

We start by reducing the subdeterminant maximization problem under a regular matroid constraint to a polynomial optimization problem as follows. 
Let $B_1,B_2,\ldots,B_m \in \mathbb{R}^d$ be the columns of $B$. 
Since $B$ is a representation of the matroid $\cM$, a set $S\subseteq M$ is a basis of $\cM$ if and only the set of the vectors $\{B_i: i\in S\}$ is linearly independent. 
Let $L=V^\top V$ be a Cholesky decomposition of the PSD matrix $L$, for $V\in \R^{d \times m}$.

Let us now consider any set $S\in \binom{[m]}{d}$ and define $I_S:= \diag{1_S}.$  
For any $S\in \binom{[m]}{d}$ we have 
$$\det\inparen{VI_S B^\top } = \det\inparen{\sum_{i\in S}{V_iB_i^\top}} = \det\inparen{V_S}\det\inparen{B^\top_S}.
$$ 
Since $B$ is a totally unimodular matrix, $|\det(B_S)|=1$ if $S\in \mathcal{B}(\cM)$ and $0$ otherwise. 
Thus for all $S\in \binom{[m]}{d}$
 \begin{equation*}
\abs{\det\left(VI_S B^\top\right)} =
\begin{cases}
|\det(V_S)| & \text{if } S\in \mathcal{B},\\
0 & \text{otherwise. } \\
\end{cases} 
\end{equation*}
\noindent
 Since for all $S\in \binom{[m]}{d}$, $\det(L_{S,S}) = \det(V_S^\top V_S)=\det(V_S)^2$, maximizing $\det(L_{S,S})$ over $S\in \mathcal{B}$ is equivalent to maximizing $\abs{f(x)}$ for $f(x):=\det(VXB^\top)$ over all the 0-1 vectors $x\in \{0,1\}^m$ subject to  $\sum_{i=1}^m x_i=d$.  
We give an approximation algorithm for this problem which proceeds in two phases.
\subsubsection*{Phase 1: Finding a Fractional Solution.}
In the first phase, we drop the  $\sum_{i=1}^m x_i=d$ condition and relax the $0-1$ condition to $x\in [0,1]^m$.
Our optimization problem then becomes
\begin{equation}\label{eq:regularcaserelaxation}
\begin{aligned}
\max_{x} ~~ & |f(x)|,&\\
\st ~~ & x\in [0,1]^m.\\
\end{aligned}
\end{equation}
\noindent
Our algorithm to find an approximate solution to~\eqref{eq:regularcaserelaxation} is as follows.
We sample a polynomial number of points $x$ from $[0,1]^m$ uniformly and independently at random. 
Then, we output the point with the largest value of $|f(x)|$. 
We analyze the performance of this algorithm in two different regimes.  

\parag{Large $d$.} 
It follows from the Cauchy-Binet formula that 
\begin{equation}\label{eq:deff}
f(x) = \sum_{S \in \cB} x^S \det(V_S) \det(B_S).
\end{equation}
Moreover, $f(x)$ is multi-affine and easy to compute (because it is just a determinant of an $m \times m$ matrix). 
We show that $|f|$ is 2-anti-concentrated. 
To this end, we show that for every $i \in [m]$ and every choice of $y_j\in [0,1], j\in[m] \setminus \{i\}$, the univariate function
$$
\tau \mapsto |f\inparen{y_1,\ldots,y_{i-1},\tau,y_{i+1},\ldots,y_m}|
$$
is 2-anti-concentrated. Such a function is of the form $\tau \mapsto |a\tau + b|$ for some $a,b\in \R$. 2-anti-concentration of such functions follows easily from Lemma~\ref{lemma:anti_conc_norm}. Indeed, by setting $d=1$ and $p=2$ in Lemma~\ref{lemma:anti_conc_norm} we obtain the 2-anti-concentration of $(\tau_1, \tau_2) \mapsto |\tau_1 a_1 + \tau_2 a_2|$, which implies our claim. 

Theorem~\ref{thm:main0} implies now that if we sample a uniform point $x$ from $[0,1]^m$ then 
$$
\pr\left[ |f(x)| > 2^{-m}(2e^2)^{-m} \cdot \OPT \right] \geq \nfrac{1}{2}.
$$
\noindent
Where $\OPT := \max_{x\in [0,1]^m} |f(x)|$ is clearly an upper bound on $\max_{S\in \cB} |\det(V_S)|$. We can amplify the probability of success by repeating the experiment several times and hence, with high probability obtain a point $\hat{x}$ such that
\begin{equation}\label{eq:appxboundregmat1}
|f(\hat{x})| > (2e)^{-2m} \cdot \OPT.
\end{equation}

\parag{Small $d$.} From \eqref{eq:deff} it is clear that the function $f$ is a polynomial of degree $d$ in $m$ variables. 
According to Theorem 2 in \cite{CW}, if we sample $x$ uniformly from the unit hypercube $[0,1]^m$, then
$$
\pr\insquare{\abs{f(x)}\leq \beta^d\cdot  \OPT} \leq C\cdot  \beta\cdot m,
$$
for any $\beta>0$ and some absolute constant $C>0$. By picking $\beta = \frac{1}{2C\cdot m}$, we conclude that with constant probability we obtain a vector $\hat{x}$ such that
\begin{equation}\label{eq:appxboundregmat2}
|f(\hat{x})| > \inparen{\frac{1}{{2mC}}}^d\cdot \OPT.
\end{equation}

\subsubsection*{Phase 2: Rounding the Fractional Solution.}
We first round $\hat{x}$ obtained in the previous phase to a $0-1$ vector, and then finally to a set $\hat{S}\in {[m] \choose d}$. 
Since $f$ is multi-affine, the restriction of $f$ to the first coordinate is a 1-dimensional affine function. 
Therefore, either 
$$
|f(0,\hat{x}_2,\ldots,\hat{x}_d)| \geq |f(\hat{x})|
~~\mbox{ 
or }
~~
|f(1,\hat{x}_2,\ldots,\hat{x}_d)| \geq |f(\hat{x})|.
$$
Hence, we can round the first coordinate without decreasing the value of $|f(\hat{x})|$, using one call to the evaluation oracle.
We proceed to the next coordinates and round them one at a time. 
Let $y\in \{0,1\}^m$ be the outcome of the above rounding algorithm. 

Let $S_0 \subseteq [m]$ such that $1_{S_0}=y$. It is likely that $|S_0|>d$, hence we will need to remove several elements from $S_0$ to obtain a set of cardinality $d$. 
Define a function $g: 2^{[m]}\to \mathbb{R}$ to be
\begin{equation*}
g(S) :=f(1_S) = \det (V_S B_S^\top).
\end{equation*}
\noindent
Note in particular that $g$ can be computed efficiently. Furthermore, by the  Cauchy-Binet formula, we have
\begin{equation} \label{CBeq}
g(S)= \sum_{T\in \binom{[m]}{d}} {g(T)} = \sum_{T\in \binom{[m]}{d}} \det(V_T) \det(B_T)
\end{equation}
\noindent 
for every subset $S\in 2^{[m]}$. 
We have $|f(y)| = |f(1_{S_0})| = |g(S_0)|$. Further, \eqref{CBeq} implies that 
$$
\sum_{i\in S_0}g(S_0 \setminus \{i\}) = (|S_0|-d) \sum_{T\in  \binom{S_0 }{d}}{g(T)} = (|S_0|-d) g(S_0).
$$
Consequently, there exists an $i\in S_0$ such that:
$$|g(S_0 \setminus \{i\})| \geq \frac{|S_0|-d}{|S_0|} |g(S_0)|.$$
In our algorithm we find such an $i$ and consider $S_1:= S_0 \setminus \{i\}$. This step of removing one element is repeated until we arrive at a set $\hat{S}\subseteq [m]$ of cardinality $d$. In this process we can guarantee that
$$|g(\hat{S})| \geq|g(S_0)|\cdot \prod_{j=1}^{|S_0|-d}\frac{j}{j+d} \geq \frac{|g(S_0)|}{{m \choose d}}.$$
Finally, since $|g(\hat{S})|=|\det(V_{\hat{S}})|$, we conclude:
$$
|\det(V_{\hat{S}})| \geq \frac{|f(y)|}{{m \choose d}}
>  \frac{1}{\binom{m}{d}} \max \inparen{(2e)^{-2m}, (2dC)^{-d}}\cdot \OPT 
$$
hence $|\det(V_{\hat{S}})| > \max \inparen{2^{-O(m)} , 2^{-O(d \log m)}}\cdot \OPT,$ and Theorem~\ref{thm:main1} follows.

\bibliographystyle{alpha}
\bibliography{references}

\newcommand{\etalchar}[1]{$^{#1}$}
\begin{thebibliography}{CDK{\etalchar{+}}17}

\bibitem[AO17]{AO17}
N.~Anari and S.~{Oveis Gharan}.
\newblock A generalization of permanent inequalities and applications in
  counting and optimization.
\newblock In {\em Proceedings of the 49th Annual ACM SIGACT Symposium on Theory
  of Computing}, STOC 2017, pages 384--396, 2017.

\bibitem[CDK{\etalchar{+}}17]{CDKSV17}
L.~E. {Celis}, A.~{Deshpande}, T.~{Kathuria}, D.~{Straszak}, and N.~K.
  {Vishnoi}.
\newblock {On the Complexity of Constrained Determinantal Point Processes}.
\newblock In {\em Approximation, Randomization and Combinatorial Optimization.
  Algorithms and Techniques, 20th International Workshop, {APPROX} 2017, and
  21st International Workshop, {RANDOM} 2017, Berkeley, CA, USA, August 16-18,
  2017.}, 2017.

\bibitem[CDKV16]{CDKV16}
L.~E. Celis, A.~Deshpande, T.~Kathuria, and N.~K. {Vishnoi}.
\newblock {How to be fair and diverse?}
\newblock {\em Fairness, Accountability, and Transparency in Machine Learning},
  2016.

\bibitem[CLL14]{CheungLL14}
Ho~Yee Cheung, Lap~Chi Lau, and Kai~Man Leung.
\newblock Algebraic algorithms for linear matroid parity problems.
\newblock {\em {ACM} Trans. Algorithms}, 10(3):10:1--10:26, 2014.

\bibitem[{\c{C}}M09]{CM09}
A.~{\c{C}}ivril and M.~Magdon{-}Ismail.
\newblock On selecting a maximum volume sub-matrix of a matrix and related
  problems.
\newblock {\em Theor. Comput. Sci.}, 410(47-49):4801--4811, 2009.

\bibitem[CTV06]{Costello06}
K.~P. Costello, T.~Tao, and V.~Vu.
\newblock Random symmetric matrices are almost surely nonsingular.
\newblock {\em Duke Math. J.}, 135(2):395--413, 11 2006.

\bibitem[CW01]{CW}
A.~Carbery and J.~Wright.
\newblock Distributional and ${L}^q$ norm inequalities for polynomials over
  convex bodies in $\mathbb{R}^n$.
\newblock {\em Mathematical research letters}, 8(3):233--248, 2001.

\bibitem[DDS16]{DeDS16}
A.~De, I.~Diakonikolas, and R.~A. Servedio.
\newblock A robust {K}hintchine inequality, and algorithms for computing
  optimal constants in fourier analysis and high-dimensional geometry.
\newblock {\em SIAM Journal on Discrete Mathematics}, 30(2):1058--1094, 2016.

\bibitem[GKL95]{GKL95}
P.~Gritzmann, V.~Klee, and D.~G. Larman.
\newblock {Largest $j$-simplices $n$-polytopes}.
\newblock {\em Discrete and Computational Geometry}, pages 477--517, 1995.

\bibitem[GT01]{GT2001}
S.~A. Goreinov and E.~E. Tyrtyshnikov.
\newblock The maximal-volume concept in approximation by low-rank matrices.
\newblock {\em Contemporary Mathematics}, 280:47--51, 2001.

\bibitem[GT17]{GT17}
R.~Gurjar and T.~Thierauf.
\newblock Linear matroid intersection is in quasi-nc.
\newblock In {\em Proceedings of the 49th Annual {ACM} {SIGACT} Symposium on
  Theory of Computing, {STOC} 2017, Montreal, QC, Canada, June 19-23, 2017},
  pages 821--830, 2017.

\bibitem[Gur06]{Gurvits06}
L.~Gurvits.
\newblock {Hyperbolic polynomials approach to Van der Waerden/Schrijver-Valiant
  like conjectures: sharper bounds, simpler proofs and algorithmic
  applications}.
\newblock In {\em Proceedings of the thirty-eighth annual ACM symposium on
  Theory of computing}, pages 417--426. ACM, 2006.

\bibitem[Har09]{Harvey09}
N.~Harvey.
\newblock Algebraic algorithms for matching and matroid problems.
\newblock {\em SIAM Journal on Computing}, 39(2):679--702, 2009.

\bibitem[Kha95]{Kha95}
L.~Khachiyan.
\newblock On the complexity of approximating extremal determinants in matrices.
\newblock {\em Journal of Complexity}, 11(1):138--153, 1995.

\bibitem[KT12]{KuleszaTaskar12}
A.~Kulesza and B.~Taskar.
\newblock {\em Determinantal Point Processes for Machine Learning}.
\newblock Now Publishers Inc., Hanover, MA, USA, 2012.

\bibitem[Lov79]{Lovasz79}
L{\'{a}}szl{\'{o}} Lov{\'{a}}sz.
\newblock On determinants, matchings, and random algorithms.
\newblock In {\em {FCT}}, pages 565--574, 1979.

\bibitem[Lov89]{Lovasz89}
L.~Lov{\'a}sz.
\newblock Singular spaces of matrices and their application in combinatorics.
\newblock {\em Boletim da Sociedade Brasileira de
  Matem{\'a}tica-Bulletin/Brazilian Mathematical Society}, 20(1):87--99, 1989.

\bibitem[{Lyo}02]{Lyons02}
R.~{Lyons}.
\newblock {Determinantal probability measures}.
\newblock {\em ArXiv Mathematics e-prints}, April 2002.

\bibitem[Nik15]{Nikolov15}
A.~Nikolov.
\newblock {Randomized Rounding for the Largest Simplex Problem}.
\newblock In {\em Proceedings of the Forty-seventh Annual ACM Symposium on
  Theory of Computing}, STOC '15, pages 861--870. ACM, 2015.

\bibitem[NS16]{NS16}
A.~Nikolov and M.~Singh.
\newblock Maximizing determinants under partition constraints.
\newblock In {\em Proceedings of the 48th Annual ACM Symposium on Theory of
  Computing}, pages 192--201, 2016.

\bibitem[O'D14]{O'Donnell2014}
R.~O'Donnell.
\newblock {\em Analysis of Boolean Functions}.
\newblock Cambridge University Press, New York, NY, USA, 2014.

\bibitem[Oxl06]{Oxley06}
J.~G Oxley.
\newblock {\em Matroid theory}, volume~3.
\newblock Oxford University Press, USA, 2006.

\bibitem[RV13]{RV13}
A.~Razborov and E.~Viola.
\newblock Real advantage.
\newblock {\em ACM Transactions on Computation Theory (TOCT)}, 5(4):17, 2013.

\bibitem[SEFM15]{DEFM14}
M.~D. Summa, F.~Eisenbrand, Y.~Faenza, and C.~Moldenhauer.
\newblock On largest volume simplices and sub-determinants.
\newblock In {\em Proceedings of the Twenty-Sixth Annual ACM-SIAM Symposium on
  Discrete Algorithms}, pages 315--323. SIAM, 2015.

\bibitem[SV17]{SV17}
D.~Straszak and N.~K. Vishnoi.
\newblock Real stable polynomials and matroids: Optimization and counting.
\newblock In {\em Proceedings of the 49th Annual ACM SIGACT Symposium on Theory
  of Computing}, STOC 2017, pages 370--383, 2017.

\end{thebibliography}

\begin{appendix}
\section{Hard Example for the Nikolov-Singh Algorithm}\label{sec:nshard}
In this section, we give an example that the algorithm proposed in~\cite{NS16} for the subdeterminant maximization under partition constraints might fail to output a set with non-zero subdeterminant with high probability, even though the expected value of the returned solution is high. 
\begin{lemma}
There exists an instance of the subdeterminant maximization problem under partition constraints, for which the optimal value is equal to $1$ and the  Algorithm~\cite{NS16} outputs a non-zero solution with exponentially small probability.  
\end{lemma}

\begin{proof}
Let $L=V^\top V$, where $V\in \R^{r\times m}$ is a matrix with $m=r^2$. The columns of $V$ are standard unit vectors $e_1, e_2, \ldots, e_r \in \R^r$ each one repeated $r$ times. We consider the problem of maximizing $\det(V_S^\top V_S)$ over sets $S\subseteq [m]$ of cardinality $r$. This is an instance of the subdeterminant maximization problem under partition constraints, when there is only one partition of size $m$ and $b_1=r$.  For such instances  the algorithm of~\cite{NS16} specializes to that of~\cite{Nikolov15}. It first solves the convex program
$$\max_{x\in P} \log \det \inparen{\sum_{i=1}^m x_i v_i v_i^\top }$$
where $P=\{x\in \R^m: \sum_{i=1}^m x_i=r, \; 0\leq x \leq 1\}$. It is not hard to see that the point 
$$z=\inparen{\nfrac{1}{r}, \nfrac{1}{r}, \ldots, \nfrac{1}{r}}\in \R^m$$
is an optimal solution to the above optimization problem. 

The output of the Nikolov-Singh algorithm is a random set $S$ sampled according to a distribution $\rho$ given by $\rho(S) \propto z^S$  for $|S|=r$.  It can be simply seen to be the uniform distribution over all subsets of $[m]$ of size $r$. 

Suppose that $\mathbf{S}$ is distributed according to $\rho$. It is immediate to see that $\det(V_{\mathbf{S}}^\top V_{\mathbf{S}})\in \{0,1\}.$ Moreover, the determinant is $1$ if and only if exactly one vector is picked from every group of $r$ copies of standard unit vectors.
$$\pr\insquare{\det(V_{\mathbf{S}}^\top V_{\mathbf{S}})=1}= \frac{r^r}{{r^2 \choose r}} \approx \frac{r^r r!}{(r^r)^2} \approx \frac{r^r \cdot r^r}{(r^r)^2 e^r}=e^{-r}.$$
In the above estimate we used Stirling approximation and ignored small polynomial factors in $r$.

The above calculation implies that with probability exponentially close to one, the randomized algorithm of~\cite{NS16} returns a trivial solution 0. To obtain a solution of value at least the expectation (which is also roughly $e^{-r}$) one needs to run this algorithm about $e^r$ times.

\end{proof}

\end{appendix}

\end{document}